\documentclass[aps,pra,onecolumn,groupedaddress,superscriptaddress,longbibliography]{revtex4-2}

\usepackage[colorlinks,citecolor=Blue,linkcolor=Red, urlcolor=Blue]{hyperref}
\usepackage{dsfont}
\usepackage{dcolumn}
\usepackage[usenames,dvipsnames]{xcolor}
\hypersetup{breaklinks=true}
\usepackage[normalem]{ulem}
\usepackage{enumerate}

\usepackage{amsthm}
\newtheorem{innercustomgeneric}{\customgenericname}
\providecommand{\customgenericname}{}
\newcommand{\newcustomtheorem}[2]{%
  \newenvironment{#1}[1]
  {%
   \renewcommand\customgenericname{#2}%
   \renewcommand\theinnercustomgeneric{##1}%
   \innercustomgeneric
  }
  {\endinnercustomgeneric}
}

\newcustomtheorem{theoremrestate}{Theorem}
\newcustomtheorem{lemmarestate}{Lemma}
\newcustomtheorem{corollaryrestate}{Corollary}

\newtheorem*{theorem*}{Theorem}
\newtheorem*{lemma*}{Lemma}
\newtheorem*{corollary*}{Corollary}

\usepackage{ifthen}

\newcommand\curRevision{1}
\newcommand{\currentrevision}[1]{
    \def\curRevision{#1}
}

\newcommand{\revdtext}[2][1]{\ifthenelse{\equal{#1}{\curRevision}}{\textcolor{blue}{#2}}{#2}}
\newcommand{\TODO}{
    \textcolor{red}{\textbf{TODO}}
    \PackageWarning{libs}{There are unresolved TODOs}
}

\newcommand{\citemissing}{\textsuperscript{\textcolor{red}{[Citation needed]}}}
\newcommand{\figmissing}{\textsuperscript{\textcolor{red}{[Figure needed]}}}
\usepackage{graphicx}
\usepackage{subcaption}
\usepackage{float}

\usepackage{tikz}
\usepackage{pgfplots}
\pgfplotsset{compat=1.18}

\usetikzlibrary{angles,quotes}
\usetikzlibrary{decorations.pathmorphing,arrows.meta}
\usepackage{bbm}
\usepackage{amsmath}
\usepackage{amssymb}
\usepackage{amsthm}
\usepackage{mathtools}

  \DeclarePairedDelimiterX\Set[1]\{\}{%
  
  #1
}

\newtheorem{theorem}{Theorem}

\newtheorem{corollary}[theorem]{Corollary}

\newtheorem{lemma}[theorem]{Lemma}

\newtheorem{definition}[theorem]{Definition}

\newcommand{\C}{\mathbb{C}}

\newcommand{\N}{\mathbb{N}}

\newcommand{\R}{\mathbb{R}}
\newcommand{\Z}{\mathbb{Z}}

\usepackage{xparse}

\newcommand{\bigO}{\mathcal{O}}

\DeclarePairedDelimiterX\pBrackets[1]{[}{]}{%

#1
}

\newcommand{\Prob}[2][]{
 \def\tmp{#1}%
   \ifx\tmp{}
     \operatorname{Pr}\pBrackets*{#2}
   \else
     \operatorname{Pr}_{#1}\pBrackets*{#2}
   \fi
}

\newcommand{\Expect}[2][]{
 \def\tmp{#1}%
   \ifx\tmp{}
     \operatorname{\mathbb{E}}\pBrackets*{#2}
   \else
     \operatorname{\mathbb{E}}_{#1}\pBrackets*{#2}
   \fi
}

\usepackage{physics}
\usepackage{braket}

\usetikzlibrary{quantikz}

\renewcommand{\ket}[1]{| #1 \rangle}
\renewcommand{\bra}[1]{\langle #1 |}
\renewcommand{\braket}[2]{\langle #1 | #2 \rangle}
\renewcommand{\proj}[1]{\ket{#1}\bra{#1}}

\newcommand{\vspan}[1]{\operatorname{span}\{ #1 \}}

\newcommand{\diag}{\operatorname{diag}}
\newcommand{\rk}{\operatorname{rk}}

\newcommand{\id}{\mathbbm{1}}
\renewcommand{\tr}{\operatorname{tr}}

\newcommand{\sL}[2]{{\mathfrak{s l}(2, \mathbb{C})}}

\currentrevision{2}

\begin{document}

\title{Quantum signal processing over SU(N)}

\author{Lorenzo Laneve}
\email{lorenzo.laneve@usi.ch}
\affiliation{Faculty of Informatics — Universit\`a della Svizzera Italiana, 6900 Lugano, Switzerland}

\begin{abstract}
  Quantum signal processing (QSP) and the quantum singular value transformation (QSVT) are pivotal tools for simplifying the development of quantum algorithms. These techniques leverage polynomial transformations on the eigenvalues or singular values of block-encoded matrices, achieved with the use of just one control qubit. \revdtext{In contexts where the control qubit is used to extract information about the eigenvalues of singular values, the amount of extractable information is limited to one bit per protocol.} In this work, we extend the original QSP ansatz by introducing multiple control qubits. \revdtext{We show that, much like in the single-qubit case, nearly any vector of polynomials can be implemented with a multi-qubit QSP ansatz, and the gate complexity scales polynomially with the dimension of such states.} \revdtext{Moreover,} assuming that powers of the matrix to transform are easily implementable -- as in Shor's algorithm -- we can achieve polynomial transformations with degrees that scale exponentially with the number of control qubits. This work aims to provide a partial characterization of the polynomials that can be implemented using this approach, with \revdtext{phase estimation schemes} and discrete logarithm serving as illustrative examples.
\end{abstract}
\maketitle

\onecolumngrid

\noindent Quantum algorithmic design is more an art than a science. To this day, quantum procedures across the literature that can provably provide complexity advantage over classical models of computation usually involve the alternations of pairs of reflections (resembling Grover's search~\cite{groverFastQuantumMechanical1996, groverQuantumMechanicsHelps1997}) or the quantum Fourier transform to efficiently expose some cyclic structure (thus generalizing Shor's algorithm~\cite{shorPolynomialTimeAlgorithmsPrime1997}). Just like in classical computation, the development of a methodology (intended as a toolbox of generally-applicable techniques) is crucial to understand how and to what extent we can achieve speed-ups in new settings. Notable examples of such techniques, which sometimes also helped to develop further intuition about already-known algorithms, include quantum walks~\cite{kempeQuantumRandomWalks2003,szegedyQuantumSpeedupMarkov2004,ambainisQuantumWalkAlgorithm2004,magniezHittingTimesQuantum2012,laneveHittingTimesGeneral2023}, span programs~\cite{reichardtSpanprogrambasedQuantumAlgorithm2012,hoyerNegativeWeightsMake2007,cornelissenSpanProgramsQuantum2020}, learning graphs~\cite{belovsApplicationsAdversaryMethod2014}, quantum divide and conquer~\cite{childsQuantumDivideConquer2022}, and quantum signal processing~\cite{lowMethodologyResonantEquiangular2016,lowOptimalHamiltonianSimulation2017a,gilyenQuantumSingularValue2019}.

In particular, quantum signal processing (QSP) is a construction that arose from the context of composite pulses in nuclear magnetic resonance~\cite{vandersypenNMRTechniquesQuantum2005}. In this setting, a parameter $z$ characterizing some quantum process is embedded in a single-qubit operator and, by intertwining this operator with a sequence of gates (independent of $z$), we can induce a polynomial transformation on $z$, thus transforming also the structure of the original underlying process.
For example, if one takes $z$ to be the eigenvalues of a unitary matrix $U$ representing some algorithm, QSP allows to carry out a uniform polynomial transformation $P$ on all the eigenvalues of $U$, i.e.,
\begin{align}
    \label{eq:quantum-eigenvalue-transform-example}
    \ket{0} \ket{\psi} \mapsto \ket{0} P(U) \ket{\psi} + \ket{1} Q(U) \ket{\psi}
\end{align}
where $Q(U)$ must be chosen in such a way that $|P|^2 + |Q|^2 \equiv 1$, in order for the state to be normalized (post-selecting on $\ket{0}$ will then enact $P(U)$). The degree of $P$ is equal to the number of steps \revdtext{(and calls to $U$)} needed for its implementation, thus marking a clear connection with the required complexity.
With some more effort, one can also use QSP to carry out a transformation on the singular values of a matrix that is \emph{block-encoded} in a unitary, e.g., a non-necessarily-square matrix in the top-left corner of a unitary (the so-called \emph{quantum singular value transformation}~\cite{gilyenQuantumSingularValue2019,tangCSGuideQuantum2023}).

Quantum eigenvalue transformation initially turned out to be a useful tool for Hamiltonian simulation~\cite{lowHamiltonianSimulationQubitization2019, lowOptimalHamiltonianSimulation2017a}, where optimal algorithms were found with simple and intuitive ideas (later also improved in~\cite{martynEfficientFullyCoherentQuantum2023}). However, the true potential of quantum signal processing was discovered with the advent of the quantum singular value transformation~\cite{gilyenQuantumSingularValue2019} (hereinafter, QSVT), where other known algorithms were found to be rewritten in terms of simple applications of the main theorem, such as fixed-point amplitude amplification~\cite{brassardQuantumAmplitudeAmplification2002,yanFixedpointObliviousQuantum2022}, linear systems solving~\cite{harrowQuantumAlgorithmLinear2009} and quantum-state preparation~\cite{mcardleQuantumStatePreparation2022,laneveRobustBlackboxQuantumstate2023}. Other applications include a QSVT-based version of phase estimation and factoring~\cite{martynGrandUnificationQuantum2021}. These results show how the QSP framework gives life to a \emph{grand unification} of many known quantum algorithms. \revdtext{The powerfulness of this framework is also studied in~\cite{rossiModularQuantumSignal2023}, where a function-first theory for devising quantum algorithms is proposed on top of the QSP/QSVT framework.}

Equation~\ref{eq:quantum-eigenvalue-transform-example} can nonetheless be seen from a different perspective: instead of `hoping' that the control qubit collapses to $\ket{0}$ so that $\ket{\psi}$ will undergo a desired (non-unitary) transformation, we could design $P, Q$ in such a way that the control qubit is the one that actually carries out useful information (indeed, this is what is done in the QSVT-based phase estimation of~\cite{martynGrandUnificationQuantum2021}, when information about the phase is extracted one bit at a time). \revdtext{With this approach in mind, the single control qubit constitutes a limitation (of one bit) on the information we can extract from the eigenvalues or singular values, forcing us to compose multiple protocols (even in non-trivial ways) to extract more information, which in turn sacrifices the simplicity of the analysis and the algorithm.} Intuitively, if we allow $b > 1$ control qubits we should design not two polynomials, but \revdtext{up to $2^b$}, and we can use the idea of QSP to extract more than one bit of information at once. In this work we show that, much like in the single-qubit case, nearly any vector of $2^b$ polynomials can be implemented with a multi-qubit QSP ansatz\revdtext{, extending previous work~\cite{haahProductDecompositionPeriodic2019, dongGroundStatePreparationEnergy2022,wangQuantumPhaseProcessing2023, motlaghGeneralizedQuantumSignal2023} for $b > 2$ control qubits. We also show that the signal processing operators can be implemented efficiently, provided that their dimension is polynomial in the input size (i.e., the number of control qubits is logarithmic in the input size). As an example, we present phase location, a problem similar to phase estimation where the eigenphase has to be localized in one of $s$ arcs of the complex circle given on input. With a fairly simple analysis, we devise a procedure that returns the correct interval with high probability, which in uneven instances turns out to be more efficient than using phase estimation with sufficient accuracy.}

\revdtext{Moreover, in situations where the unitary $U$ to transform is such that the powers $\{ U^k \}_k$ are efficiently implementable (more than naively repeating $U$ $k$ times), one can use such powers to achieve polynomials of $U$ with higher degree in fewer steps. Although this cannot be always the case (a general scheme for efficiently implementing powers of a unitary is excluded by the so-called \emph{no-fast forwarding theorem}~\cite{atiaFastforwardingHamiltoniansExponentially2017}), there exist unitaries that satisfy this condition.}

As an example, we show how (a reinterpretation of) Shor's algorithm for the discrete logarithm~\cite{dewolfQuantumComputingLecture2019} emerges by itself from a special case of multi-qubit quantum signal processing, adding more algorithms under the QSP framework. \revdtext{However, we do not need this assumption for these results to be useful: we also derive the original phase estimation algorithm, as well as a more robust variant involving a gaussian window.}

We start in Section~\ref{sec:one-qubit} by reviewing a less-known variant of single-qubit QSP~\cite{motlaghGeneralizedQuantumSignal2023}, which will be simpler to understand for the multiple-qubit case, (yet equivalent to the more traditional pictures~\cite{martynGrandUnificationQuantum2021,haahProductDecompositionPeriodic2019}), and we provide a different proof for it. \revdtext{The results are then extended to multiple qubits in Section~\ref{sec:multi-qubit-qsp}, where we show the result for an arbitrary number of qubits and we apply it to solve phase location. Section~\ref{sec:exponential-steps} discusses the case where we have access to $U^k$, and gives a partial characterization for the achievable polynomials, namely in the case where the length $n = 1$, giving necessary and sufficient conditions, and using this result to (re-)derive phase estimation schemes as well as computing discrete logarithms. Section~\ref{sec:general-case} gives final remarks characterizations about the general case $n,b > 1$ with exponential step, as well as a general outline of the proofs. Appendix~\ref{apx:polynomial-unitary-approximation} shows how to efficiently implement a unitary $U$ of length $d$, scheme that is used to implement the signal processing operators whenever they are computed numerically, while Appendices~\ref{apx:triangular-matrices},~\ref{apx:qsp-multi-qubit-proof} give the formal arguments for all the results.}
\section{Review of single-qubit quantum signal processing}
\label{sec:one-qubit}
\noindent Here we show a slightly different formulation of quantum signal processing over $SU(2)$ (construction also proposed in~\cite{motlaghGeneralizedQuantumSignal2023}) that can be later extended to multiple qubits. We define the \emph{signal operator} $\Tilde{w}$ as
\begin{align*}
    \Tilde{w} =
    \begin{bmatrix}
        1 & 0 \\
        0 & z
    \end{bmatrix}.
\end{align*}
This simply represents a controlled unitary $cU = \proj{0} \otimes \id + \proj{1} \otimes U$, restricted to a generic eigenspace associated with the eigenvalue $z \in U(1)$. It is known that intertwining this signal operator with a sequence of $SU(2)$ elements (the so-called \emph{signal processing operators}) gives all the possible (analytic) polynomials.
\begin{theorem}[\cite{motlaghGeneralizedQuantumSignal2023}]
    \label{thm:single-qubit-qsp-analytic}
    Let $P, Q \in \C[z]$. There exist $A_0, \ldots, A_n \in SU(2)$ such that
    \begin{align*}
        A_n \Tilde{w} A_{n-1} \Tilde{w} \cdots \Tilde{w} A_0 \ket{0} =
        \begin{bmatrix}
            P(z) \\
            Q(z)
        \end{bmatrix}
    \end{align*}
    if and only if
    \begin{enumerate}[(i)]
        \item $\deg P, Q \le n$;
        \item $|P|^2 + |Q|^2 = 1$ for every $z \in U(1)$.
    \end{enumerate}
\end{theorem}
\begin{figure}
    \centering
    \input{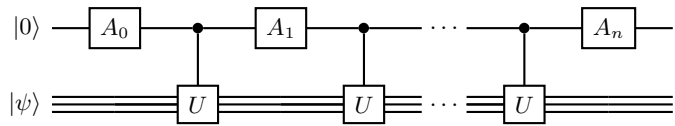}
    \caption{Ansatz for the single-qubit quantum eigenvalue transform, according to Theorem~\ref{thm:single-qubit-qsp-analytic}. If $U\ket{\psi} = z \ket{\psi}$, then we constructed the transformation $\ket{\psi} \rightarrow P(z) \ket{0} \ket{\psi} + Q(z) \ket{1} \ket{\psi}$, i.e., we are implementing a superposition of $P(U), Q(U)$.}
    \label{fig:one-qubit-qsp-ansatz}
\end{figure}

\noindent It is worth pointing out that $A_1, \ldots, A_n$ need not be arbitrary $SU(2)$ unitaries (but $A_0$ does): Since the signal operator is expressed in the $Z$-basis, we can, for example, restrict ourselves to operators $A$ such that $\tr(AX) = 0$ or $\tr(AY) = 0$ without hindering the generality of the theorem. This can also be seen in~\cite{motlaghGeneralizedQuantumSignal2023} where a simple trick (involving the $ZXZ$-decomposition) allows to parameterize the operators with only two angles, and not three. The circuit carrying out the so-called \emph{quantum eigenvalue transform} leveraging Theorem~\ref{thm:single-qubit-qsp-analytic} is depicted in Figure~\ref{fig:one-qubit-qsp-ansatz}~\cite{dongGroundStatePreparationEnergy2022,wangQuantumPhaseProcessing2023}. It is important to remark that a better-known version of QSP is a direct corollary.
\begin{corollary}
    \label{thm:single-qubit-qsp-laurent}
    Let $P, Q \in \C[z, z^{-1}]$ (Laurent polynomials in $z$). There exist $A_0, \ldots, A_n \in SU(2)$ such that
    \begin{align*}
        A_n \begin{bmatrix}
            z^{-1} & 0 \\
            0 & z
        \end{bmatrix}
        A_{n-1}
        \begin{bmatrix}
            z^{-1} & 0 \\
            0 & z
        \end{bmatrix}
        \cdots
        \begin{bmatrix}
            z^{-1} & 0 \\
            0 & z
        \end{bmatrix}
        A_0 \ket{0} =
        \begin{bmatrix}
            P(z) \\
            Q(z)
        \end{bmatrix}
    \end{align*}
    if and only if
    \begin{enumerate}[(i)]
        \item $\deg P, Q \le n$;
        \item $|P|^2 + |Q|^2 = 1$ for every $z \in U(1)$;
        \item $P, Q$ have parity $n \bmod 2$.
    \end{enumerate}
\end{corollary}
\noindent This follows from the fact that $\diag(z^{-1}, z) = z^{-1} \diag(1, z^2)$, which causes a shift of the coefficients, from $\{-n, \ldots, n\}$ to $\{0, \ldots, 2n\}$, while the parity condition follows from the fact that Theorem~\ref{thm:single-qubit-qsp-analytic} is applied with $z^2$ instead of $z$. \revdtext{Indeed, by this argument if one is allowed access to both $\Tilde{w}, \Tilde{w}^\dag$ (i.e., to both $U, U^\dag$), it is possible to construct any Laurent polynomial (even of indefinite parity), by simply implementing a $2n$-degree analytic polynomial, and then shifting all the coefficients by calling $(U^\dag)^n$, and this operation will at most double the number of calls to the unitary.} We refer to the construction of Corollary~\ref{thm:single-qubit-qsp-laurent} as the \emph{Laurent picture}~\cite{dongGroundStatePreparationEnergy2022, haahProductDecompositionPeriodic2019,chaoFindingAnglesQuantum2020,wangQuantumPhaseProcessing2023, rossiMultivariableQuantumSignal2022}, while we use the term \emph{analytic picture} to denote the construction of Theorem~\ref{thm:single-qubit-qsp-analytic}~\cite{motlaghGeneralizedQuantumSignal2023}. \revdtext{For the reader more familiar with the traditional \emph{Chebyshev picture}~\cite{martynGrandUnificationQuantum2021} (from which the quantum singular value transformation can be obtained~\cite{gilyenQuantumSingularValue2019,tangCSGuideQuantum2023}) the connection can be made by setting $x = (z^{-1} + z)/2 \in [-1, 1]$, thus having:
\begin{align*}
    H \begin{bmatrix}
        z^{-1} & 0 \\
        0 & z
    \end{bmatrix} H & = 
    \begin{bmatrix}
        x & i \sqrt{1 - x^2} \\
        i \sqrt{1 - x^2} & x
    \end{bmatrix}
\end{align*}
where $H$ is the Hadamard gate.}
\section{\revdtext{Quantum signal processing over multiple qubits}}
\label{sec:multi-qubit-qsp}
\noindent We now consider the case where the control qubits are $b > 1$. Consider the following signal operator
\begin{align}
    \label{eq:multi-qubit-linear-signal-operator}
    \Tilde{w} = \diag(1, \ldots, 1, z, \ldots, z) = \diag(1, z) \otimes \id.
\end{align}
\noindent This operator only requires a single controlled version of the unitary $U$ whose transformation we want to carry out. Our first result shows that Theorem~\ref{thm:single-qubit-qsp-analytic} has a natural extension to multiple qubits.
\begin{theorem}
    \label{thm:multi-qubit-qsp-analytic-linear}
    Let $\{ P_x \}_{0 \le x < 2^b} \in \C[z]$. There exist $A_0, \ldots, A_n \in SU(2^b)$ such that:
    \begin{align*}
        A_n \Tilde{w} A_{n-1} \Tilde{w} \cdots \Tilde{w} A_0 \ket{0} = \sum_x P_x(z) \ket{x}
    \end{align*}
    if and only if:
    \begin{enumerate}[(i)]
        \item $\deg P_x \le n$ for every $0 \le x < 2^b$;
        \item $\sum_x |P_x(z)|^2 = 1$ for every $z \in U(1)$.
    \end{enumerate}
\end{theorem}
\noindent Theorem~\ref{thm:multi-qubit-qsp-analytic-linear} tells that, just like in the single-qubit case, any normalized vector of $2^b$ polynomials can be implemented with $n$ steps, where $n$ is the maximum degree of any of the polynomials. For convenience, we introduce some definitions that will be useful throughout the work: a \emph{polynomial state} $\ket{\gamma(z)} = \sum_k \ket{\gamma_k} z^k$ is a parameterized vector that is a valid quantum state for every choice of $z \in U(1)$, i.e., $\braket{\gamma(z)}{\gamma(z)} \equiv 1$. The degree of $\ket{\gamma}$ is the highest $k$ such that $\ket{\gamma_k} \neq 0$. Thus, Theorem~\ref{thm:multi-qubit-qsp-analytic-linear} tells us that any polynomial state of degree $n$ (of any dimension) can be achieved with a $n$-length single-qubit QSP ansatz. Curiously enough, by looking at the proof one can see that the result does not depend on how many 1's and $z$'s are present in the diagonal of Equation~\ref{eq:multi-qubit-linear-signal-operator} (as soon as we have least one of each). This suggests that the non-uniqueness of the $A_k$'s implementing a given polynomial state is even more severe than in the single-qubit case. 

An important question arises about the efficient implementation of the signal processing operators $A_k$ which, for $b \gg 1$, is not even guaranteed to exist. For any constant $b$ (i.e., not depending on our input), $A_k$ has an efficient implementation by the Solovay-Kitaev theorem~\cite{nielsenQuantumComputationQuantum2010a,childsLectureNotesQuantum}. When $d = 2^b$ is polynomial in the input size, then it is possible to implement any unitary matrix (this involved a non-trivial construction, see Appendix~\ref{apx:polynomial-unitary-approximation}), thus concluding that QSP over polynomial dimension is implementable in polynomial time.

We state here some corollaries and additional results that will be useful in practice: if we consider the Laurent signal operator $\Tilde{v} = \diag(z^{-1}, z) \otimes \id$, we obtain a version of multi-qubit QSP in the Laurent picture.
\begin{corollary}
    \label{thm:multi-qubit-qsp-laurent-linear}
    Let $\{ P_x \}_{0 \le x < 2^b} \in \C[z, z^{-1}]$. There exist $A_0, \ldots, A_n \in SU(2^b)$ such that:
    \begin{align*}
        A_n \Tilde{v} A_{n-1} \Tilde{v} \cdots \Tilde{v} A_0 \ket{0} = \sum_x P_x(z) \ket{x}
    \end{align*}
    if and only if:
    \begin{enumerate}[(i)]
        \item $\deg P_x \le n$ for every $0 \le x < 2^b$;
        \item $\sum_x |P_x(z)|^2 = 1$ for every $z \in U(1)$;
        \item All $P_x$ have parity $n \bmod 2$.
    \end{enumerate}
\end{corollary}
\noindent This corollary is proven exactly like Corollary~\ref{thm:single-qubit-qsp-laurent}. Moreover, it is important to remark that, as in the single-qubit case, if we are allowed to use both $\Tilde{w}, \Tilde{w}^\dag$, we can remove condition (iii) by simply doubling the number of steps. The second result we extend from the single-qubit case involves completion of polynomial states.
\begin{theorem}
    \label{thm:polynomial-state-completion}
    Let $P_x(z)$ be $2^b - 1$ Laurent polynomials of degree at most $n$ such that $\sum_{x = 0}^{2^b - 2} |P_x|^2 \le 1$ on the unit circle. Then there exists a polynomial $P_{2^b}(z)$ such that $\sum_{x = 0}^{2^b - 1} |P_x(z)|^2 = 1$. The degree of this polynomial will be at most $n$.
\end{theorem}
\begin{proof}
    The polynomial $R(z) = 1 - \sum_{x = 0}^{2^b - 2} |P_x(z)|^2$ is non-negative on the unit circle by assumption. Therefore, by the Fej\'er-Riesz theorem~\cite{geronimoPositiveExtensionsFejerRiesz2004, hussenFejerRieszTheoremIts2021}, there exists an analytic polynomial $w(z)$ such that:
    \begin{align*}
        |w(z)|^2 = R(z)
    \end{align*}
    Since $R(z)$ is a Laurent polynomial of degree at most $2n$, then the $w(z)$ will have degree at most $2n$, and $P_{2^b - 1}(z) = z^{-n} w(z)$ will be a Laurent polynomial of degree at most $n$ satisfying the same condition.
\end{proof}

\subsection{\revdtext{Example: phase location}}
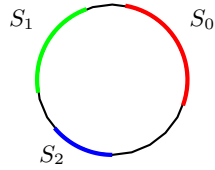
\begin{figure}
    \centering
    \begin{tikzpicture}
    \draw [black,thick,domain=0:360] plot ({cos(\x)}, {sin(\x)});
    \draw [red,ultra thick,domain=-20:80] plot ({cos(\x)}, {sin(\x)});
    \draw [blue,ultra thick,domain=220:270] plot ({cos(\x)}, {sin(\x)});
    \draw [green,ultra thick,domain=110:190] plot ({cos(\x)}, {sin(\x)});

    \node at (1.2, 0.8) {$S_0$};
    \node at (-1.2, 0.8) {$S_1$};
    \node at (-0.8, -1) {$S_2$};
\end{tikzpicture}
    \caption{Setting of the phase location problem. The input specifies $s$ arcs of the complex circle, with the promise that they are separated by `neutral zones' (the black arcs) of length $\ge \Delta$. The objective is to locate the eigenphase $\varphi$ in one of these intervals.}
    \label{fig:phase-location}
\end{figure}

\noindent Let $\{ S_j \}_j$ be a set of $s$ disjoint (and connected) subsets of the unitary circle $U(1)$. We are also guaranteed that these subsets are distant at least $\Delta$, i.e., if $x \in S_i, y \in S_j$ with $i \neq j$ then $|x - y|_{2\pi} \ge \Delta$~(see Figure~\ref{fig:phase-location}). Given a unitary $U$ as oracle and an eigenstate $\ket{\psi}$ with eigenvalue $e^{2\pi i \varphi}$, with the promise that $\varphi \in S_j$ for some $j$, we want to return the subset containing $\varphi$. Unlike in phase estimation, we do not want to recover full information about $\varphi$, but the subsets $S_j$ are not equal in size in general. Our final goal is to implement a vector of $k$ polynomials $P_j$, such that $|P_j|^2 \ge 1 - \epsilon$ within $S_j$. We use the following polynomials:
\begin{align*}
    R_{k, n}(z) = - \frac{2k e^{-k^2/2}}{\sqrt{\pi}} \left( I_0(k^2/2) \frac{z - z^{-1}}{2i} + \sum_{j=1}^{(n-1)/2} I_j(k^2/2) \left( \frac{z^{2j+1} - z^{-2j-1}}{2i (2j+1)} + \frac{z^{-2j-1} + z^{-(2j-1)}}{2i (2j-1)} \right) \right)
\end{align*}
where $I_k$ are the modified Bessel functions of the first kind.
\begin{theorem}[Corollary 6.27 in \cite{lowQuantumSignalProcessing2017}]
    \label{thm:sign-function-approximation}
    For a degree $n = \Theta(k \log \frac{1}{\epsilon})$ we have
    \begin{align*}
        \left|R_{k,n}(z) + \erf\left(k \frac{z - z^{-1}}{2i}\right)\right| \le \epsilon.
    \end{align*}
    As a direct corollary, for $k = \Theta(\frac{1}{\Delta} \log \frac{1}{\epsilon})$ we have
    \begin{align*}
        \left|R_{k,n}(e^{i 2\pi \varphi}) + \operatorname{sgn}\left(\sin 2\pi \varphi\right)\right| \le \epsilon \,\,\,\,\, \text{for $|\sin \varphi| \ge \Delta$}
    \end{align*}
\end{theorem}
\noindent Theorem~\ref{thm:sign-function-approximation} tells us that $R_{k,n}(z)$ approximates a square wave, being $-1$ in the upper half of the complex circle and $+1$ in the lower half (see Figure~\ref{fig:square-wave-approx}, left).
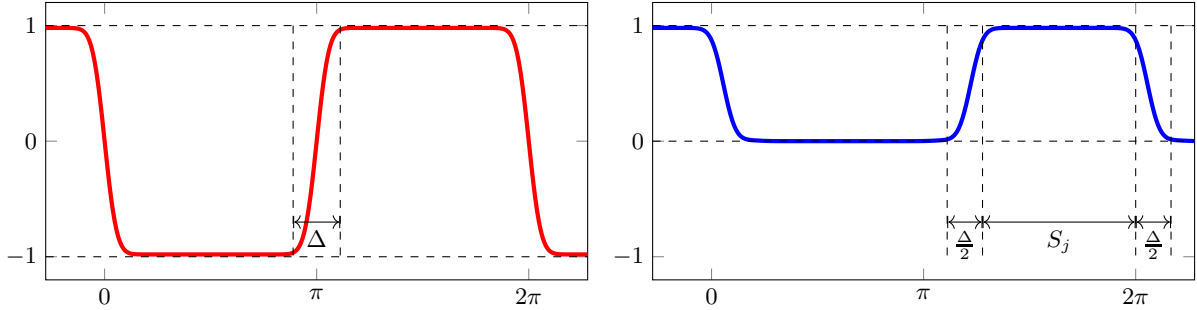
\begin{figure*}
    \centering

    \makeatletter
    \pgfmathdeclarefunction{erf}{1}{%
    \begingroup
        \pgfmathparse{#1 > 0 ? 1 : -1}%
        \edef\sign{\pgfmathresult}%
        \pgfmathparse{abs(#1)}%
        \edef\x{\pgfmathresult}%
        \pgfmathparse{1/(1+0.3275911*\x)}%
        \edef\t{\pgfmathresult}%
        \pgfmathparse{%
        1 - (((((1.061405429*\t -1.453152027)*\t) + 1.421413741)*\t 
        -0.284496736)*\t + 0.254829592)*\t*exp(-(\x*\x))}%
        \edef\y{\pgfmathresult}%
        \pgfmathparse{(\sign)*\y}%
        \pgfmath@smuggleone\pgfmathresult%
    \endgroup
    }
    \makeatother

    \begin{tikzpicture}
\begin{axis}[
    width=250pt,height=150pt,
    xmin=-50,xmax=410,
    ymin=-1.2,ymax=1.2,
    samples=50,
    xtick={0,180,360},
    xticklabels={$0$, $\pi$, $2\pi$},
    ytick={-1,0,1},
    yticklabels={$-1$, $0$, $1$},
    grid style={line width=.1pt, draw=gray!10},
    axis line style={latex-latex}]

    \addplot[red, ultra thick, domain=-50:410, samples=300] (x, {-0.98*erf(5*sin(x))});

    \draw [dashed] (axis cs:{-50},1) -- (axis cs:{410},1);
    \draw [dashed] (axis cs:{-50},-1) -- (axis cs:{410},-1);

    \draw[|<->|] (axis cs:{160},-0.7) -- node[below] {$\Delta$} (axis cs:{200},-0.7);

    \draw [dashed] (axis cs:{160},1) -- (axis cs:{160},-1);
    \draw [dashed] (axis cs:{200},1) -- (axis cs:{200},-1);
\end{axis}
\end{tikzpicture}
    \begin{tikzpicture}
\begin{axis}[
    width=250pt,height=150pt,
    xmin=-50,xmax=410,
    ymin=-1.2,ymax=1.2,
    samples=50,
    xtick={0,180,360},
    xticklabels={$0$, $\pi$, $2\pi$},
    ytick={-1,0,1},
    yticklabels={$-1$, $0$, $1$},
    grid style={line width=.1pt, draw=gray!10},
    axis line style={latex-latex}]

    \addplot[blue, ultra thick, domain=-50:410, samples=300] (x, {(1-0.98*erf(5*sin(x-10)))/2 * (1-0.98*erf(5*sin(x-40)))/2});

    \draw [dashed] (axis cs:{-50},1) -- (axis cs:{410},1);
    \draw [dashed] (axis cs:{-50},0) -- (axis cs:{410},0);

    \draw[|<->|] (axis cs:{200},-0.7) -- node[below] {$\frac{\Delta}{2}$} (axis cs:{230},-0.7);
    \draw [dashed] (axis cs:{200},1) -- (axis cs:{200},-1);
    \draw [dashed] (axis cs:{230},1) -- (axis cs:{230},-1);

    \draw[|<->|] (axis cs:{360},-0.7) -- node[below] {$\frac{\Delta}{2}$} (axis cs:{390},-0.7);
    \draw [dashed] (axis cs:{360},1) -- (axis cs:{360},-1);
    \draw [dashed] (axis cs:{390},1) -- (axis cs:{390},-1);

    \draw[|<->|] (axis cs:{230},-0.7) -- node[below] {$S_j$} (axis cs:{360},-0.7);


  
\end{axis}
\end{tikzpicture}
    \caption{On the left, the error function wave (as per Theorem~\ref{thm:sign-function-approximation}), approximated by the polynomial $R_{k,n}$. On the horizontal axis we have the eigenphase $2\pi \varphi$. This approximates well the ${-}1/{+}1$ square wave everywhere except in the $\Delta$-neighbourhoods of the integer multiples of $\pi$. On the right, the construction of $P_j$ for a given arc $S_j$, as per Eq.~(\ref{eq:phase-location-polynomial-construction}): we ensure that $P_j \neq 1$ for the whole $S_j$, and we use $\Delta/2$ more space (guaranteed by the promise on the input), to do the smooth transition between $1$ and $0$.}
    \label{fig:square-wave-approx}
\end{figure*}
Notice that this is the same polynomial used in amplitude amplification and phase estimation algorithms based on QSVT~\cite{martynGrandUnificationQuantum2021}. We then construct the polynomial as follows:
\begin{align}
    \label{eq:phase-location-polynomial-construction}
    P_j(z) = \frac{\sqrt{1 - \delta}}{1 + \epsilon} \cdot \frac{1 + R_{k,n}(z e^{-i a})}{2} \cdot \frac{1 + R_{k,n}(- z e^{-i (a + b)})}{2}.
\end{align}
This consists in the multiplication of two square waves oscillating between $0, 1$ (the factor $\frac{\sqrt{1 - \delta}}{2+\epsilon}$ will be needed to ensure $\sum_j |P_j|^2 \le 1$ while retaining the $\epsilon$ approximation error). $P_j$ will approximate a square wave being $+1$ on the interval $[a, b]$, and $0$ everywhere else (see Figure~\ref{fig:square-wave-approx}, right). Actually, Eq.~(\ref{eq:phase-location-polynomial-construction}) only gives polynomials $P_j(z)$ being $1$ for at most half of the period, i.e., the construction works only if $S_j$ covers at most half of the circle. Fortunately, at most one interval can cover more than half of the circle (which we call $S_0$ without loss of generality). We design $P_j(z)$ for all the other $S_j$, and then we find $P_0$ by completion using Theorem~\ref{thm:polynomial-state-completion}. By construction, $P_j$ satisfies the following:
\begin{align*}
    |P_j(z)|^2 & \le 1 - \delta \\
    |P_j(z) - 1| & \le \delta + 2\epsilon & \text{for $z \in S_j$} \\
    |P_j(z)| & \le \epsilon & \text{for $z \not\in S_j \pm \Delta/2$}
\end{align*}
We take $\delta = \epsilon'/4, \epsilon = \epsilon'/4s$, so that $|P_j|^2$ is away from $0$ or $1$ up to $\epsilon'$ and
\begin{align*}
    \sum_{j \neq 0} |P_j(z)|^2 & \le 1 - \frac{\epsilon'}{4} + (s - 1) \frac{\epsilon'}{4s} \le 1
\end{align*}
thus satisfying the condition for completion. Since $P_j(z)$ has degree $\bigO(\frac{1}{\Delta} \log\frac{1}{\epsilon})$, we can construct the polynomial state with $\bigO(\frac{1}{\Delta} \log \frac{s}{\epsilon'})$ QSP steps, i.e., $\bigO(\frac{s^3}{\Delta} \log \frac{s}{\epsilon'})$ total gates and $\bigO(s)$ additional qubits to obtain a $\epsilon'$ failure probability. Notice that this is significantly better than approaches using phase estimation in cases where the lengths of $S_j$'s are uneven. We remark that this problem could be tackled in principle with single-qubit QSP, but this needs a separate QSP protocol for each of the intervals, making the analysis (in particular, the error bound) more involved. Moreover, using a similar argument one can achieve an algorithm for phase estimation similar to the QSVT-based one in~\cite{martynGrandUnificationQuantum2021}, where instead of a binary search, we carry out a $2^b$-ary search, extracting $b$ bits of the phase at each step. This does not, however, improve the complexity of the original algorithm (unsurprisingly, since we would surpass the Heisenberg limit).

\section{\revdtext{Exponential steps}}
\label{sec:exponential-steps}

\noindent We now consider a different signal operator on $b$ qubits:
\begin{align*}
    \Tilde{w}_b = \diag(1, z^{2^{b-1}}) \otimes \cdots \otimes \diag(1, z^2) \otimes \diag(1, z) = \diag(1, z, z^2, \ldots, z^{2^b - 1})
\end{align*}
i.e., a diagonal containing the first $2^b$ powers of the eigenvalue $z$. By the first decomposition it is clear that we can implement $\Tilde{w}_b$ with controlled powers of $U^{2^k}$. This is particularly useful when $U^{2^k}$ is efficiently implementable (more than merely repeating calls to $U$), as the efficient implementations of high powers are a shortcut to high-degree polynomials (one call to $U^k$ already gives a $k$-degree polynomial). It is important to remark that this assumption surely does not hold for every unitary $U$: the no fast-forwarding theorem~\cite{atiaFastforwardingHamiltoniansExponentially2017} excludes such a general implementation scheme. There are, however, particular unitaries exploiting some algebraic structure, for example:
\begin{itemize}
    \item Let $f(x): \Z_N \rightarrow \Z_M$ be an efficiently computable function, and let $U$ be its phase kickback oracle, i.e.,
    \begin{align*}
        U \ket{x} = e^{2\pi i f(x)/M} \ket{x}
    \end{align*}
    This unitary is efficiently implementable with two copies of the circuit for $f$, and a quantum Fourier transform over $\Z_M$. An implementation of $U^{2^k}$ consists in replacing the circuit for $f(x)$ with a circuit for $2^k \cdot f(x)$ (i.e., shift the bits of the output by $k$).

    \item Let $G$ be a (non-necessarily abelian) group, whose operation $(x, y) \mapsto x \cdot y$ is efficiently computable. The unitary $U: \ket{x} \mapsto \ket{a \cdot x}$ is efficiently computable (here, $a \in G$ and $\ket{x}$ is intended to be some unique encoding of $x \in G$). Moreover, the unitary $U^{2^k}$ realizes the transformation
    \begin{align*}
        U^{2^k} \ket{x} = \ket{a^{2^k} \cdot x}
    \end{align*}
    which is also efficiently computable, by simply replacing $a$ with $a^{2^k}$ (encoded classically) in the circuit for $U$.
\end{itemize}
The second example, in particular, is the main property for the efficiency of algorithms solving the Hidden Subgroup Problem~\cite{childsLectureNotesQuantum,dewolfQuantumComputingLecture2019}.

Unfortunately, characterizing implementable polynomial states with this ansatz is arguably hard, since it is a special case of (a multi-qubit extension of) multivariable quantum signal processing (M-QSP)~\cite{rossiMultivariableQuantumSignal2022}. If the degree of a polynomial state does not exceed its dimension, we can nonetheless give a full characterization:
\begin{theorem}
    \label{thm:multi-qubit-qsp-analytic-one-length}
    Let $\ket{\gamma(z)} = \sum_k \ket{\gamma_k} z^k$ be a polynomial state of degree $\le 2^b - 1$. There exist $A_0, A_1 \in SU(2^b)$ such that
    \begin{align*}
        A_1 \Tilde{w}_b A_0 \ket{0} = \ket{\gamma(z)}
    \end{align*}
    if and only if $\{ \ket{\gamma_k} \}_k$ are pairwise orthogonal.
\end{theorem}
\noindent The proof of this claim can be found in Appendix~\ref{apx:qsp-multi-qubit-proof}, and it shows that $A_1$ is a unitary containing the normalized $\ket{\gamma_k}$ as columns, while $A_0$ prepares the state
\begin{align*}
    A_0 \ket{0} = \sum_k \sqrt{\braket{\gamma_k}{\gamma_k}} \ket{k}.
\end{align*}
In other words, given an analytic form of $\ket{\gamma(z)}$, Theorem~\ref{thm:multi-qubit-qsp-analytic-one-length} yields an analytic form for $A_0, A_1$. Notice the difference with general-length QSP, where the signal processing operators must be estimated numerically. The availability of a simple analytic form allows us to express elements of $SU(2^b)$ and, in many cases, a simple circuit is available to implement such transformations. We conclude by stressing that, although the existence of these `fast-forwardable' unitaries is an important motivation for this work, algorithm design can benefit from Theorem~\ref{thm:multi-qubit-qsp-analytic-one-length} even when it is necessary to call $U$ for $2^b$ times, as we shall see in the next section.

\subsection{Original phase estimation}
\label{sec:phase-estimation}
\noindent Phase estimation is an important problem in quantum computation: given a unitary $U$ (along with its powers $U^{2^k}$) as a black-box, and an eigenstate $\ket{\psi}$ such that $U \ket{\psi} = e^{2\pi i \varphi} \ket{\psi}$ for $\varphi \in [0, 1)$, it provides a $b$-bit approximation of $2^b \varphi$ with high probability. While algorithms for phase estimation leveraging QSVT are already known~\cite{martynGrandUnificationQuantum2021} (they essentially decompose the quantum Fourier transform into a binary search), here we want to show how the original phase estimation circuit comes out by itself with a simple argument involving Theorem~\ref{thm:multi-qubit-qsp-analytic-one-length}.

\noindent Assuming that $2^b \varphi \in \N$ for simplicity \revdtext{(a more robust algorithm is shown in the next section)}, in order to devise an algorithm for phase estimation using multiple-qubit QSP, we need to design a family of polynomials $\{ P_x(z) \}_x$ such that:
\begin{align*}
    P_x(e^{2\pi i \varphi}) = \delta_{x, 2^b \varphi} = \frac{1}{2^b} \sum_{k = 0}^{2^b - 1} \omega_{2^b}^{k (2^b \varphi - x)}
\end{align*}
which, by replacing $z = e^{2\pi i \varphi}$ becomes:
\begin{align*}
    P_x(z) = \frac{1}{2^b} \sum_{k = 0}^{2^b - 1} \omega_{2^b}^{-xk} z^k
\end{align*}
where $\delta_{\cdot,\cdot}$ is the Kronecker delta, and we simply used the expansion using the roots of unity $\omega_{2^b}^y$. By grouping the $2^b$ polynomials into the polynomial state $\ket{\gamma(z)}$, the coefficient vectors are
\begin{align*}
    \ket{\gamma_k} = \frac{1}{2^b} \sum_{y = 0}^{2^b - 1} \omega_{2^b}^{-xk} \ket{x}
\end{align*}
As $\braket{\gamma_k}{\gamma_\ell} = 0$ for $k \neq \ell$, following the proof of Theorem~\ref{thm:multi-qubit-qsp-analytic-one-length} we find $A_1, A_0$, where $A_1$ turns out to be the inverse quantum Fourier transform and $A_0$ is any unitary mapping $\ket{0}$ to the $b$-fold Hadamard state $\ket{+}^{\otimes b}$. This shows that phase estimation in its original form~\cite[p.~221]{nielsenQuantumComputationQuantum2010a} can be expressed in terms of quantum signal processing, where the presence of the quantum Fourier transform is a consequence (and not a cause) of the expansion of the Kronecker delta.

\subsection{\revdtext{Phase estimation with gaussian window}}
\noindent Here we show more robust algorithm for phase estimation, using tools from~\cite{chi-fangQuantumThermalState2023}. Consider the following gaussian wave function:
\begin{align*}
    \hat{\phi}(s) = \frac{1}{\sqrt[4]{2\pi \sigma^2}} e^{-s^2/4\sigma^2}
\end{align*}
\noindent It is well-known that the (inverse) Fourier transform is still a Gaussian~\cite{schumacherQuantumProcessesSystems2010}, more precisely:
\begin{align*}
    \phi(x) = \sqrt[4]{\frac{2\sigma^2}{\pi}} e^{-\sigma^2 x^2}
\end{align*}
\noindent In order to `encode' a gaussian in a quantum state, we introduce the following concept:
\begin{definition}
    \label{def:periodic-wrapping}
    Let $f : \R \rightarrow \C$ be a square-integrable function. Given $T > 0$, the $T$-wrapping of $f$ is defined as:
    \begin{align*}
        f^{(T)}(x) = \sum_{n \in \Z} f(x + nT)
    \end{align*}
\end{definition}
\noindent The concept of periodic wrapping introduced by Definition~\ref{def:periodic-wrapping} is simple: the first case we use is a $2\pi$-wrapping of the Gaussian, which can be seen as a wrapping around the complex circle where the phase to estimate lies~(see Figure~\ref{fig:gaussian-wrapping}).
\begin{figure}
    \centering
    \pgfmathdeclarefunction{gauss}{3}{%
  \pgfmathparse{1/(#3*sqrt(2*pi))*exp(-((#1-#2)^2)/(2*#3^2))}%
}

\begin{tikzpicture}
    \draw [black,thick,domain=0:360] plot ({cos(\x)}, {sin(\x)});

    \draw [black,ultra thick,domain=-180:180, samples=200] plot ({(1 + 3*gauss(\x/10, 7,1))*cos(\x)}, {(1 + 3*gauss(\x/10, 7,1))*sin(\x)});

    \draw [red,ultra thick,domain=-180:180, samples=200] plot ({(1 + 3*gauss(\x/10, 7,1))*cos(\x)}, {(1 + 3*gauss(\x/10, 7,1))*sin(\x)});

    \draw [black,thick,domain=0:360, samples=200] plot ({\x/80-10}, 0);
    \draw [red,ultra thick,domain=0:360, samples=200] plot ({\x/80-10}, {3*gauss(\x/10, 7,1)});
    \node at (-9.12, 0) {\textbullet};
    \node at (-9.12, -0.3) {$2\pi\varphi$};

    \node at (-10, 0) {\textbullet};
    \node at (-10, -0.3) {$0$};

    \node at ({36/8-10}, 0) {\textbullet};
    \node at ({36/8-10}, -0.3) {$2\pi$};

    \node at (0, 0) {\textbullet};
    \node at ({0.99*cos(69)}, {0.99*sin(69)}) {\textbullet};
    \node at ({0.75*cos(69)}, {0.75*sin(69)}) {$2\pi\varphi$};

    \node at (-3, 0) {$\Longrightarrow$};
\end{tikzpicture}
    \caption{Visualization of the periodic wrapping of a gaussian $\hat{\phi}(x - 2\pi\varphi)$ around the complex circle. The gaussian peaks at $x = 2\pi\varphi$. The $2\pi$-wrapping $\hat{\phi}^{(2\pi)}$ can be seen as the plot of $\hat{\phi}$ (on the whole real line) being wrapped around the complex circle. By the properties of the gaussian, we can take only a fixed period of length $2\pi$ centered in the peak, so that we can neglect the rest of the plot when doing the wrapping.}
    \label{fig:gaussian-wrapping}
\end{figure}
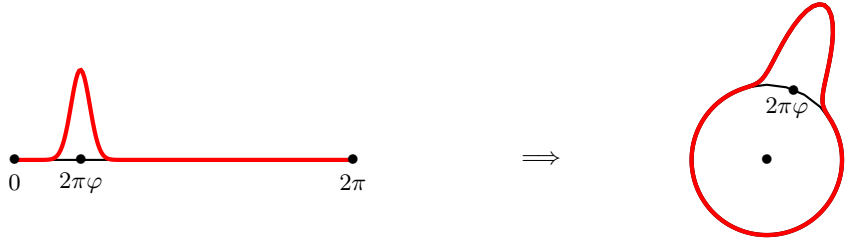
The ultimate goal will be to construct the state:
\begin{align}
    \ket{\bar{\phi}} \propto \sum_{k = 0}^{N - 1} \hat{\phi}^{(2\pi)}\left(\frac{2\pi k}{N} - 2\pi\varphi\right) \ket{k} \label{eq:final-gaussian-state}
\end{align}
\noindent where the (discretized) wrapped gaussian will peak when $k/N$ is around the true phase $\varphi$, even when $N\varphi$ is not an integer. The (continuous) inverse Fourier transform of $\hat{\phi}(s - 2\pi\varphi)$ is $\phi(x) e^{2\pi i x\varphi} = \phi(x) z^x$. This fact allows us to rewrite the state in (\ref{eq:final-gaussian-state}) as a polynomial state in the eigenvalue $z$. The only obstacle to this reasoning is the difference between discrete and continuous Fourier transform, which we overcome with the following result:
\begin{theorem}[Fact A.1 in \cite{chi-fangQuantumThermalState2023}, simplified]
    \label{thm:ft-discrete-continuous-duality}
    Let $f$ be a square integrable function, let $N$ be a positive integer, and fix two periods $T, W$ such that $TW = 2\pi N$. By defining $\hat{f}(p) = \frac{1}{\sqrt{2\pi}} \int_{\R} f(x) e^{-ipx} dx$ to be the continuous Fourier transform for $f$, the following holds:
    \begin{align*}
        \sqrt{T} \left[\frac{1}{\sqrt{N}} \sum_{y \in \Z_N} f^{(T)}\left( \frac{yT}{N} \right) \omega_N^{-ky}\right] = \sqrt{W} \left[ \hat{f}^{(W)} \left( \frac{k W}{N} \right) \right]
    \end{align*}
    for every $k \in \Z_N$.
\end{theorem}
\noindent By applying Theorem~\ref{thm:ft-discrete-continuous-duality} we rewrite the state in Eq.~(\ref{eq:final-gaussian-state}):
\begin{align*}
    \sum_{k = 0}^{N - 1} \hat{\phi}^{(2\pi)}\left(\frac{2\pi k}{N} - 2\pi \varphi\right) \ket{k} & = \frac{1}{\sqrt{2\pi}} \sum_{k \in \Z_N} \sum_{y \in \Z_N} [\phi(y) z^y]^{(N)} \cdot \omega_N^{-ky} \ket{k} \\
    & = \frac{1}{\sqrt{2\pi}} \sum_{k \in \Z_N} \sum_{y \in \Z_N} \sum_{n \in \Z} \phi(y + nN) z^{y + nN} \cdot \omega_N^{-ky} \ket{k}.
\end{align*}
\noindent The above state involves all the powers of $z$ by the definition of periodic wrapping, but we can neglect all the terms that are too far away from the peak of the gaussian, giving us the state (after renormalization)
\begin{align*}
    \frac{1}{K \sqrt{N}} \sum_{k = 0}^{N - 1} \sum_{y = -N/2}^{N/2 - 1} \phi(y) z^y \cdot \omega_N^{-ky} \ket{k}
\end{align*}
where $K^2 = \sum_k \phi^2(k)$ is the normalization factor. More formally, the truncated state is only $\bigO(\sigma e^{-N^2 \sigma^2 / 2})$ far from the initial state in $2$-norm (by a gaussian tail bound). We can multiply by the global phase $z^{N/2}$ to obtain a polynomial state which satisfies the conditions of Theorem~\ref{thm:multi-qubit-qsp-analytic-one-length} (indeed, we can also implement the Laurent polynomial). We find that $A_1$ is again the quantum Fourier transform over $\Z_N$, while $A_0 \ket{0}$ has to construct the truncated gaussian state
\begin{align*}
    A_0 \ket{0} = \frac{1}{K} \sum_{y = -N/2}^{N/2 - 1} \phi(y) \ket{y}
\end{align*}
whose construction is already known to be feasible efficiently~\cite{mcardleQuantumStatePreparation2022}. This argument can in principle be used directly with other prior states, such as cosine~\cite{rendonEffectsCosineTapering2022} and Kaiser windows~\cite{greenawayCaseStudyQSVT2024, berryAnalyzingProspectsQuantum2024}.

\subsection{Discrete logarithm}
\noindent In the discrete logarithm problem, we are given a cyclic group $G = \langle g \rangle$ of order $N$, whose operation $(a, b) \mapsto a \cdot b$ we can easily compute. Given some group element $r \in G$, find the smallest integer $\ell$ such that $r = g^\ell$. Here we show how to recover (a variant of) Shor's algorithm for the discrete logarithm using multi-qubit QSP. We consider the following two unitaries:
\begin{align*}
    U \ket{a} & = \ket{g \cdot a} \\
    V \ket{a} & = \ket{r \cdot a} = U^\ell \ket{a}
\end{align*}
where $\ket{a}$ is represented as a $b$-bit string. Note that by a repeated square trick $U^{2^k}, V^{2^k}$ are efficiently computable. While there may be some $b$-bit strings not representing elements of $G$, we only consider the subspace spanned by the group elements from now on. By Lagrange's theorem, $U^N = \id$, and thus the eigenvalues of $U$ are the roots of unity $\omega_N^s$, for $s \in \Z_N$. Let $\ket{u_s}$ be the eigenstate associated with the eigenvalue $\omega_N^s$ and fix an arbitrary one for now. The goal is to find two values $x = s\ell, y = s$, so that we can compute the discrete logarithm $\ell = x \cdot y^{-1}$.

Thus we will have polynomials $P_{x,y}(z)$ indexed by two values. In order to obtain useful information about the discrete logarithm we design the following polynomials:
\begin{align*}
    P_{x,y}(\omega_N^s) & = \delta_{x, s\ell} \cdot \delta_{y, s} \\
    & = \left( \sum_{k = 0}^{N - 1} \omega_N^{k(s\ell - x)} \right) \cdot \left( \sum_{k = 0}^{N - 1} \omega_N^{k(s - y)} \right)
\end{align*}
and once again, by replacing $z = e^{2\pi i \varphi}$:
\begin{align*}
    P_{x,y}(z) & = \left( \sum_{k = 0}^{N - 1} \omega_N^{-kx} (z^\ell)^k \right) \cdot \left( \sum_{k = 0}^{N - 1} \omega_N^{-ky} z^k \right)
\end{align*}
Since these polynomials can be written as product of two separate families of polynomials (one involving $x$, and one involving $y$), this can be implemented through a tensor product of two QSP protocols using $U$ and $V = U^\ell$, respectively (indeed they are phase estimation polynomials so we already know that they satisfy the conditions of Theorem~\ref{thm:multi-qubit-qsp-analytic-one-length}). By noticing that
\begin{align*}
    \ket{u_s} = \frac{1}{\sqrt{N}} \sum_{k \in \Z_N} \omega_N^{-sk} \ket{g^k} \Longrightarrow \ket{1} = \frac{1}{\sqrt{N}} \sum_{s \in \Z_N} \ket{u_s}
\end{align*}
i.e., $\ket{1}$ is an equal superposition of $\ket{u_s}$, and we conclude that by inputting the identity element $\ket{1}$ we obtain $(s\ell, s)$ for a uniformly random $s$, which has an inverse in $\Z_N$ with probability $\Omega(\phi(N)/N) = \Omega(1/\log \log N)$ ($\phi(N)$ being Euler's totient function). The derived algorithm is depicted in the circuit of Figure~\ref{fig:discrete-log-circuit}.
\begin{figure}
    \centering
    \input{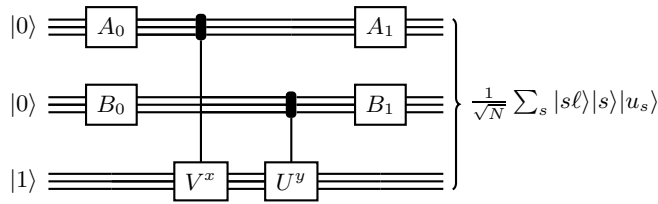}
    \caption{Circuit for the discrete logarithm over order-$N$ cyclic group $G = \langle g \rangle$ using two parallel QSP protocols. Here we are given $r = g^\ell$, and we want to find $\ell$. Here $V$ implements $a \mapsto y \cdot a$ and $U$ implements $a \mapsto g \cdot a$. Following Theorem~\ref{thm:multi-qubit-qsp-analytic-one-length} we find that $A_1, B_1$ are both inverse quantum Fourier transforms over $\Z_N$, while $A_0, B_0$ prepare the equal superposition over the elements of $G$.}
    \label{fig:discrete-log-circuit}
\end{figure}

\section{\revdtext{Remarks on the general case}}
\label{sec:general-case}

\noindent We give some final remarks on the general case $n, b > 1$ with exponential steps, to set up a path for future works. The definitions here are also relevant for the proofs of the results, which are summarized in Figure~\ref{fig:proof-diagram} and formally shown in Appendix~\ref{apx:qsp-multi-qubit-proof}.
\begin{definition}
    For a $m$-dimensional polynomial state $\ket{\gamma(z)} = \sum_k \ket{\gamma_k} z^k$, we define the square coefficient matrix $\Gamma_j$, whose columns are the $m$ coefficients vectors $\ket{\gamma_j}, \ket{\gamma_{j+1}}, \ldots, \ket{\gamma_{j+m-1}}$.
\end{definition}
\noindent Notice that if we apply a unitary $A$ (independent of $z$) to the polynomial state $\ket{\gamma(z)}$, by linearity we are simultaneously transforming all the coefficient vectors $\ket{\gamma_k} \mapsto A\ket{\gamma_k}$, and thus also $\Gamma_j \mapsto A \Gamma_j$. The crucial step in proving feasibility of construction for a given polynomial state $\ket{\gamma(z)}$ is to look for a so-called \emph{reduction}.
\begin{definition}
    A $m$-dimensional polynomial state $\ket{\gamma'(z)}$ is said to be a $\Tilde{w}$-reduction of the state $\ket{\gamma(z)}$ if there exists $A \in SU(m)$ such that
    \begin{align*}
        \ket{\gamma'(z)} = (A \Tilde{w})^\dag \ket{\gamma(z)}.
    \end{align*}
\end{definition}
\noindent This condition can be rewritten as $\ket{\gamma} = A \Tilde{w} \ket{\gamma'}$, i.e., $\ket{\gamma}$ can be constructed from $\ket{\gamma'}$ with an additional QSP step. Clearly we obtain a much different definition if we replace the single-step operator $\Tilde{w}$ with the one with exponential steps $\Tilde{w}_b$ defined in Section~\ref{sec:exponential-steps}.

In order to decompose a given $\ket{\gamma}$ into a QSP protocol, we want to look for reductions of lower degree. Each reduction will lower the degree until a zero-degree polynomial (i.e. a quantum state) is reached. Thus $A_0$ will produce this state, and the rest of the protocol is simply given by the reduction steps in reverse order. A first question is: when does a polynomial state of degree $n$ admit a reduction of lower degree? For this we have a definite and clear answer
\begin{lemma}
    \label{thm:exponential-reduction}
    The polynomial state $\ket{\gamma(z)}$ of degree $n$ has a reduction
    \begin{align*}
        \ket{\gamma'(z)} = (A \Tilde{w}_b)^\dag \ket{\gamma(z)}
    \end{align*}
    of degree $\le n - (2^b - 1)$ for some $A \in SU(2^b)$ if and only if $\Gamma_0^\dag \Gamma_{n - (2^b - 1)}$ is lower triangular.
\end{lemma}
\noindent Notice that from this lemma, we immediately deduce Theorem~\ref{thm:multi-qubit-qsp-analytic-one-length}, since $\Gamma_{n - (2^b - 1)} = \Gamma_0$, and $\Gamma_0^\dag \Gamma_0$ being lower triangular means that the columns $\{ \ket{\gamma_k} \}$ are pairwise orthogonal. A stronger result is possible if we are only looking for a reduction of one degree.
\begin{lemma}[Linear reduction]
    \label{thm:linear-reduction}
    Any polynomial state $\ket{\gamma(z)}$ of degree $n$ admits a reduction of degree $\le n - 1$.
\end{lemma}
\noindent This result yields Theorems~\ref{thm:single-qubit-qsp-analytic} and~\ref{thm:multi-qubit-qsp-analytic-linear}, since $n$ reductions are sufficient to reduce any degree-$n$ polynomial state to a quantum state. The reason why Theorem~\ref{thm:multi-qubit-qsp-analytic-one-length} is limited compared to Theorem~\ref{thm:multi-qubit-qsp-analytic-linear} is because a polynomial state whose degree was exponentially reduced does not necessarily admit another exponential reduction (i.e., the condition of Lemma~\ref{thm:exponential-reduction} is not necessarily preserved by a reduction). Understanding which polynomials can be exponentially reduced multiple times remains an open question.
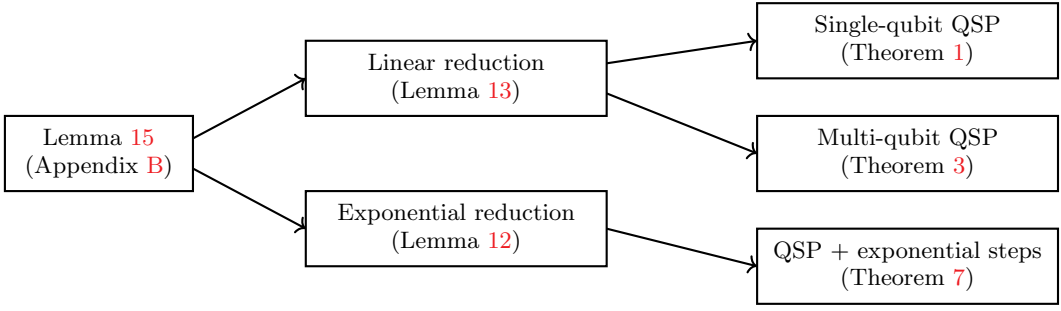
\begin{figure}
    \centering
    \begin{tikzpicture}
    \draw[align=center, thick] (0,0) rectangle (2.5,1) node[pos=.5] {Lemma~\ref{thm:simultaneous-triangularization}\\(Appendix~\ref{apx:triangular-matrices})};

    \draw[->, thick] (2.5, 0.3) -- (4, -0.5);
    \draw[->, thick] (2.5, 0.7) -- (4, 1.5);

    \draw[align=center,thick] (4,1) rectangle (8,2) node[pos=.5] {Linear reduction\\(Lemma~\ref{thm:linear-reduction})};
    \draw[align=center,thick] (4,-1) rectangle (8,0) node[pos=.5] {Exponential reduction\\(Lemma~\ref{thm:exponential-reduction})};

    \draw[->, thick] (8, 1.7) -- (10, 2);
    \draw[->, thick] (8, 1.3) -- (10, 0.5);
    \draw[->, thick] (8, -0.5) -- (10, -1);

    \draw[align=center,thick] (10,1.5) rectangle (14,2.5) node[pos=.5] {Single-qubit QSP\\(Theorem~\ref{thm:single-qubit-qsp-analytic})};
    \draw[align=center,thick] (10,0) rectangle (14,1) node[pos=.5] {Multi-qubit QSP\\(Theorem~\ref{thm:multi-qubit-qsp-analytic-linear})};
    \draw[align=center,thick] (10,-1.5) rectangle (14,-0.5) node[pos=.5] {QSP + exponential steps\\(Theorem~\ref{thm:multi-qubit-qsp-analytic-one-length})};
\end{tikzpicture}
    \caption{Diagram of the results in this work. Every theorem is proven using Lemma~\ref{thm:simultaneous-triangularization}, a result about simultaneously triangularizing two matrices, in Appendix~\ref{apx:triangular-matrices}. This allows to derive Theorems~\ref{thm:exponential-reduction},~\ref{thm:linear-reduction}.}
    \label{fig:proof-diagram}
\end{figure}

\section{\revdtext{Discussion}}
\label{sec:discussion}
\noindent In this work we considered variants and extensions of quantum signal processing involving multiple control qubits. The first claim shows that the result available in the single-qubit setting (stating that every vector of polynomials can be implemented) carries over to any dimension. As an example we showed how a problem called phase location can be solved by devising an algorithm in a convenient way, by using fewer gates than merely using phase estimation with sufficient accuracy. We then constructed a different signal operator $\Tilde{w}_b = \diag(1, z, z^2, \ldots, z^{2^b - 1})$ and characterized the polynomials obtainable with a single step of the new ansatz. This allowed us to (re-)derive phase estimation algorithms, with uniform and gaussian window states. In cases when the unitary $U$ to transform is easy to fast-forward (i.e., the cost of $U^k$ scales logarithmically with $k$), this allows a one-length QSP protocol to reach exponentially high degrees. A crucial example is Shor's algorithm and the class of HSP algorithms~\cite{childsLectureNotesQuantum}, whose special case for the discrete logarithm is explicitly shown. The algorithms are thus developed with a function-first approach, i.e., starting from the family of polynomials that we want to achieve (in the spirit of~\cite{rossiModularQuantumSignal2023}): as a consequence, the quantum Fourier transform in the examples of Section~\ref{sec:exponential-steps} emerged from the expansion of the Kronecker delta (and not the other way around).

The hardness of characterizing the arising polynomials for this ansatz is also due to the tight relation with multivariable quantum signal processing (M-QSP)~\cite{rossiMultivariableQuantumSignal2022}: indeed both ideas have a common extension (in the Laurent picture) involving multiple qubits and a signal operator embedding multiple variables such as
\begin{align*}
    \Tilde{w}(a, b) =
    \begin{bmatrix}
        a^{-1} & 0 \\
        0 & a
    \end{bmatrix}
    \otimes
    \begin{bmatrix}
        b^{-1} & 0 \\
        0 & b
    \end{bmatrix}
\end{align*}
while in~\cite{rossiMultivariableQuantumSignal2022} one forces all the control qubits (but one) to be classical, here we are giving up the simplicity of the single control qubit, but we return to the univariate case by making the restriction $a = z, b = z^2$. It is a possible future work to see if one can exploit the result on triangular matrices in Appendix~\ref{apx:triangular-matrices} to characterize implementable multivariable polynomials.

It is also important to remark that, while the QSVT~\cite{gilyenQuantumSingularValue2019} is usually referred to as the `multiple-qubit extension of quantum signal processing', it is fundamentally different from our construction: the focus of QSVT is usually the transformation of the block-encoded operator, and the control qubit is mainly used as a success flag indicating that the correct block is picked out (usually one is only interested in one of them, and the complementing polynomial is derived using the Fej{\'e}r-Riesz theorem~\cite{hussenFejerRieszTheoremIts2021}).

Perhaps a multiple-qubit quantum singular value transformation is possible (where singular values of a block-encoded matrix are transformed simultaneously by the $2^b$ polynomials in each subspace), but it requires a definition of multi-qubit QSP in the Chebyshev picture, which does not seem to be unique.

We conclude by highlighting a limitation of multi-qubit QSP: while $SU(2)$ operations are easy to implement, implementing general $SU(2^b)$ elements is not always possible efficiently when $b$ is super-logarithmic in the input. We nonetheless show that such protocol is efficient on polynomial dimensions, which is of theoretical interest. On the other hand, when using an ansatz with exponential steps, Theorem~\ref{thm:multi-qubit-qsp-analytic-one-length} allows to retrieve an analytical expression for the two signal processing operators, which in some cases allows to understand them enough to find efficient circuits to implement them, even on exponential dimension.
It is an open question, however, to characterize sub-algebras of the $b$-qubit algebra like in~\cite{chaoFindingAnglesQuantum2020}, containing only efficiently computable operators even for polynomial $b$.

\section*{Acknowledgements}
\noindent I would like to thank William Schober and Stefan Wolf for useful feedback, and Yuki Ito for a correction on the proof of Lemma~\ref{thm:simultaneous-triangularization}. LL acknowledges support from the Swiss National Science Foundation (SNSF), project No.~\texttt{200020\_214808}.

\bibliography{refs}

\appendix
\section{Approximating unitary transformations of polynomial size}
\label{apx:polynomial-unitary-approximation}

\noindent In order for multi-qubit QSP to be practically useful, one needs a systematic way to implement the signal processing operators $A_k$. Not surprisingly, approximating a unitary transformation from $SU(d)$ is much more difficult when $d > 2$. The Solovay-Kitaev theorem~\cite{nielsenQuantumComputationQuantum2010a} can certainly be useful when $d = \bigO(1)$, but when the number of qubits involved is polynomial in the input size, this problem is more general than quantum state preparation, which requires a number of steps polynomial in $d$ (i.e., exponential in the number of qubits)~\cite{sandersBlackBoxQuantumState2019,bauschFastBlackBoxQuantum2022,laneveRobustBlackboxQuantumstate2023}. When the dimension $d$ itself is polynomial in the input, however, we can use QSVT to implement any unitary $U$ efficiently.
\begin{theorem}
    Let $U$ be a $d \times d$ unitary matrix stored in a classical computer. The transformation $U$ can be implemented up to error $\epsilon$ in $2$-norm by a quantum circuit using $\bigO(d^3 \log(1/\epsilon))$
    gates.
\end{theorem}
\begin{proof}
    We first classically compute $H$ with $\lVert H \rVert \le 1$ such that $U = e^{iH}$. We implement a block encoding of $H$ using encoding of sparse matrices~\cite[Lemma 48]{gilyenQuantumSingularValue2019}, so that implementing a block encoding of $H$ takes $\bigO(d^3)$ steps (this includes a round of amplitude amplification). We then implement $U$ using a fully-coherent algorithm for Hamiltonian simulation with $t = 1$~\cite{martynEfficientFullyCoherentQuantum2023}.
\end{proof}
\noindent We remark that the unitaries implemented with these methods are not exact, as we introduce an $\epsilon$ error. Nonetheless, each implementation only introduces an additive factor to the overall error in the QSP protocol, so by replacing $\epsilon \leftarrow \epsilon/(n+1)$ in the implementation of each signal processing operator, we will be able to construct the given polynomial state up to $\epsilon$ error.
\section{Triangular matrices}
\label{apx:triangular-matrices}

\noindent It is known that, if two matrices $A, B$ are both upper (lower) triangular, then also $AB$ is upper (lower) triangular. Here we are interested in the converse: given that $AB$ is upper (lower) triangular, can we conclude that $A$ and $B$ are upper (lower) triangular? The answer is yes, up to a common unitary transformation.
\begin{lemma}
    \label{thm:simultaneous-triangularization}
    Let $A, B$ be two $N \times N$ matrices. Then the product $A^\dag B$ is lower triangular if and only if there exists a unitary $Q$ such that $Q^\dag A$ is upper triangular and $Q^\dag B$ is lower triangular.
\end{lemma}
\noindent It is already known that for any square matrix $A$ there always exists a unitary $Q$ such that $Q^\dag A$ is upper triangular (the so-called $QR$-decomposition~\cite{trefethenNumericalLinearAlgebra2022}), and a similar decomposition for lower triangular matrices. Lemma~\ref{thm:simultaneous-triangularization} gives a necessary and sufficient condition for two matrices having a common triangularizing matrix, which is similar in spirit to the fact that two matrices share an eigenbasis if and only if they commute.
\begin{proof}
    We first prove the easy direction: if such $Q$ exists then
    \begin{align*}
        A^\dag B = A^\dag Q Q^\dag B = (Q^\dag A)^\dag Q^\dag B
    \end{align*}
    which is a product of two lower triangular matrices.
    
    For the converse, let $\ket{a_k}, \ket{b_k}$ be the columns of $A, B$, respectively. We proceed by induction on $N$: the claim is trivial for $N = 1$, as any scalar is trivially both upper and lower triangular. \revdtext{For $N = 2$, $A^\dag B$ being lower triangular means that $\braket{a_0}{b_1} = 0$, and choosing $\ket{u_0} \propto \ket{a_0}, \ket{u_1} \propto \ket{b_1}$ as the columns of $Q$ will obtain the desired triangularization} (possibly completing the basis if zero vectors arise). For $N > 2$, we fix the first and the last column of $Q$, as in the case $N = 2$:
    \begin{align*}
        \ket{u_0} \propto \ket{a_0},\ \  \ket{u_{N-1}} \propto \ket{b_{N-1}}.
    \end{align*}
    Notice that these two vectors are orthogonal, since $A^\dag B$ is lower triangular. Moreover,
    \begin{align*}
        \bra{N-1} Q^\dag A \ket{j} & = \braket{u_{N-1}}{a_j} \propto \braket{b_{N-1}}{a_j} = 0 & \text{for $j < N - 1$} \\
        \bra{0} Q^\dag B \ket{j} & = \braket{u_0}{b_j} \propto \braket{a_0}{b_j} = 0 & \text{for $j > 0$} \\
        \bra{i} Q^\dag A \ket{0} & = \braket{u_i}{a_0} \propto \braket{u_i}{u_0} = 0 & \text{for $i > 0$} \\
        \bra{i} Q^\dag B \ket{N - 1} & = \braket{u_i}{b_{N-1}} \propto \braket{u_i}{u_{N-1}} = 0 & \text{for $i < N - 1$}
    \end{align*}
    The first two statements also follow from the fact that $A^\dag B$ is lower triangular, except in two limit cases: if only $\ket{a_0} = 0$, we will take $\ket{u_0}$ to be an arbitrary vector orthogonal to $\ket{b_j}$ for every $j > 0$ in order to enforce the second condition (and analogously if only $\ket{b_{N-1}} = 0$). \revdtext{If $\ket{a_0} = \ket{b_{N-1}} = 0$, we can take $\ket{u_0}$ as above, while we choose $\ket{u_{N-1}}$ to be an arbitrary vector orthogonal to both $\vspan{\ket{a_j} : j < N - 1}$ and $\ket{u_0}$ (this is possible because $\ket{a_0} = 0$, so $\vspan{\ket{a_j} : j < N - 1}$ has dimension at most $N - 2$)}. The other two conditions hold as long as we will choose the other columns of $Q$ to be orthogonal to $\ket{u_0}, \ket{u_{N-1}}$. By only choosing these two columns of $Q$, the \revdtext{four} conditions above fix the `frame' of the two matrices, i.e.,
    \begin{align*}
        Q^\dag A =
        \begin{bmatrix}
            * & \cdots & * & * \\
            0 & & & * \\
            \vdots & & &  \vdots \\
            0 & \cdots & 0 & *
        \end{bmatrix},
        \ \ \
        Q^\dag B =
        \begin{bmatrix}
            * & 0 & \cdots & 0 \\
            \vdots & & & \vdots \\
            * & & & 0 \\
            * & * & \cdots & *
        \end{bmatrix}.
    \end{align*}
    The inner blocks of the two matrices are of size $N-2 \times N-2$. We complete the unitary $Q$ with arbitrary columns and we call the inner blocks \revdtext{in this representation} $A', B'$. Now notice that, for $0 < i < j < N - 1$:
    \begin{align*}
        0 = \bra{i} A^\dag B \ket{j} = \bra{i} A^\dag Q Q^\dag B \ket{j} =
        \begin{bmatrix}
            * \\ A' \ket{i-1} \\ 0
        \end{bmatrix}^\dag
        \begin{bmatrix}
            0 \\ B' \ket{j-1} \\ *
        \end{bmatrix}
        = \bra{i-1} A'^\dag B' \ket{j-1}
    \end{align*}
    implying that $A'^\dag B'$ is lower triangular and, by induction, there exists a $N - 2 \times N - 2$ unitary matrix \revdtext{$Q'$} that triangularizes $A', B'$ as desired. The claim then holds since the unitary
    \begin{align*}
        Q
        \begin{bmatrix}
            1 &  & \\
            & Q' & \\
            & & 1
        \end{bmatrix}
    \end{align*}
    triangularizes $A, B$, as desired.

\end{proof}

\section{Proofs for general multi-qubit quantum signal processing}
\label{apx:qsp-multi-qubit-proof}

\noindent We want to use Lemma~\ref{thm:simultaneous-triangularization} that we proved in Appendix~\ref{apx:triangular-matrices} to derive results for multiple-qubit quantum signal processing. Following Figure~\ref{fig:proof-diagram}, the initial claim, from which we easily derive the rest of the results, is given by
\begin{lemmarestate}{\ref{thm:exponential-reduction}}[Exponential reduction]
    The polynomial state $\ket{\gamma(z)}$ of degree $n$ has a reduction
    \begin{align*}
        \ket{\gamma'(z)} = (A \Tilde{w}_b)^\dag \ket{\gamma(z)}
    \end{align*}
    of degree $\le n - (2^b - 1)$ for some $A \in SU(2^b)$ if and only if $\Gamma_0^\dag \Gamma_{n - (2^b - 1)}$ is lower triangular.
\end{lemmarestate}
\begin{proof}
    By Lemma~\ref{thm:simultaneous-triangularization}, $\Gamma_0^\dag \Gamma_{n - (2^b - 1)}$ is lower triangular if and only if there exists a unitary $A$ such that $A^\dag \Gamma_0$ is upper triangular and $A^\dag \Gamma_{n - (2^b - 1)}$ is lower triangular. Since $A^\dag \ket{\gamma(z)} = \sum_k A^\dag \ket{\gamma_k} z^k$, then we also transform $\Gamma_0 \rightarrow A^\dag \Gamma_0, \Gamma_{n - (2^b - 1)} \rightarrow A^\dag \Gamma_{n - (2^b - 1)}$ when we apply $A^\dag$ to $\ket{\gamma}$. Moreover, applying $\Tilde{w}_b^\dag$ transforms the polynomial entries:
    \begin{align*}
        \bra{x} \Tilde{w}^\dag_b A^\dag \ket{\gamma} & = z^{-x} \bra{x} A^\dag \ket{\gamma}
    \end{align*}
    i.e., $\Tilde{w}^\dag_b$ shifts the coefficient of the $x$-th polynomial by $x$ degrees to the left. If the first (last) coefficients of $\ket{\gamma}$ form an upper (lower) triangular matrix, this shift will result in all the new polynomials having degree $\le n - (2^b - 1)$ (see Figure~\ref{fig:two-qubit-shift} for a visual explanation in the case $b = 2$).
\end{proof}
\begin{figure}
    \centering
    \begin{tikzpicture}[scale=1.7]
    \definecolor{badsquares}{rgb}{0.49, 0.98, 1.0}

    \draw[decorate,decoration={brace,amplitude=4pt}] 
    (0,1) -- (2,1);
    \node at (1,1.3){$A^\dag \Gamma_0$};
    \draw[decorate,decoration={brace,amplitude=4pt}] 
    (4.5,1) -- (6.5,1);
    \node at (5.5,1.3){$A^\dag \Gamma_{n - 3}$};

    \node [anchor=center] at (-2, 0.75) {$A^\dag \ket{\gamma(z)}$};
    \node [anchor=center] at (-2, 0.25) {$P$};
    \node [anchor=center] at (-2, -0.25) {$Q$};
    \node [anchor=center] at (-2, -0.75) {$R$};
    \node [anchor=center] at (-2, -1.25) {$S$};

    \draw [dashed, anchor=south] (0,-1.7) -- (0,0.7) node {$0$};
    \draw [dashed, anchor=south] (5,-1.7) -- (5,0.7) node {$n - 3$};
    
    \foreach \x in {0, ..., 3} {
        \foreach \k in {0, ..., 3} {
            \def\sqColor{\ifnum\x>\k badsquares\else white\fi}
            \def\polyLetter{\symbol{\numexpr112+\x}}

            \draw[fill=\sqColor] ({0.5*\k}, {-0.5*\x}) rectangle ({0.5*\k+0.5}, {-0.5*\x+0.5}) node[midway] {$\polyLetter_\k$};
        }
    }

    \foreach \x in {0, ..., 3} {
        \foreach \k [evaluate=\k as \h using {int(3-\k)}] in {0, ..., 3} {
            \def\sqColor{\ifnum\x<\k badsquares\else white\fi}
            \def\polyLetter{\symbol{\numexpr112+\x}}

            \draw[fill=\sqColor] ({4.5+0.5*\k}, {-0.5*\x}) rectangle ({4.5+0.5*\k+0.5}, {-0.5*\x+0.5}) node[midway] {$\polyLetter_{\ifnum\h=0 n \else{n - \h}\fi}$};
        }
    }

    \node [anchor=center] at (3.25, -0.5) {...};
    \node [anchor=center] at (2.5, -2) {\large ${\downarrow}\ \  \Tilde{w}_2^\dag$};

\begin{scope}[shift={(0,-3)}]
    \node [anchor=center] at (-2, 0.75) {$\Tilde{w}_2^\dag A^\dag \ket{\gamma(z)}$};
    \node [anchor=center] at (-2, 0.25) {$P$};
    \node [anchor=center] at (-2, -0.25) {$Q$};
    \node [anchor=center] at (-2, -0.75) {$R$};
    \node [anchor=center] at (-2, -1.25) {$S$};

    \draw [dashed, anchor=south] (0,-1.7) -- (0,0.7) node {$0$};
    \draw [dashed, anchor=south] (5,-1.7) -- (5,0.7) node {$n - 3$};
    
    \foreach \x in {0, ..., 3} {
        \foreach \k in {0, ..., 3} {
            \def\sqColor{\ifnum\x>\k badsquares\else white\fi}
            \def\polyLetter{\symbol{\numexpr112+\x}}

            \draw[fill=\sqColor] ({0.5*\k-0.5*\x}, {-0.5*\x}) rectangle ({0.5*\k+0.5-0.5*\x}, {-0.5*\x+0.5}) node[midway] {$\polyLetter_\k$};
        }
    }

    \foreach \x in {0, ..., 3} {
        \foreach \k [evaluate=\k as \h using {int(3-\k)}] in {0, ..., 3} {
            \def\sqColor{\ifnum\x<\k badsquares\else white\fi}
            \def\polyLetter{\symbol{\numexpr112+\x}}

            \draw[fill=\sqColor] ({4.5+0.5*\k-0.5*\x}, {-0.5*\x}) rectangle ({4.5+0.5*\k+0.5-0.5*\x}, {-0.5*\x+0.5}) node[midway] {$\polyLetter_{\ifnum\h=0 n \else{n - \h}\fi}$};
        }
    }

    \node [anchor=center] at (2.5, -0.5) {...};
\end{scope}
\end{tikzpicture}
    \caption{Action of $\Tilde{w}_2^\dag$ on the coefficients of the four polynomials, where each column corresponds to a power of $z$. For example, since $\Tilde{w}_b^\dag$ maps $R(z) \rightarrow z^{-2} R(z)$, all the coefficients of $R$ are shifted by two positions to the left. If the coefficient matrices $A^\dag \Gamma_0$ and $A^\dag \Gamma_{n - (2^b - 1)}$ turn out to be upper and lower triangular, respectively, then all the coefficients colored in the figure will be $0$, and the resulting polynomial state will have degree $\le n - 3$.}
    \label{fig:two-qubit-shift}
\end{figure}
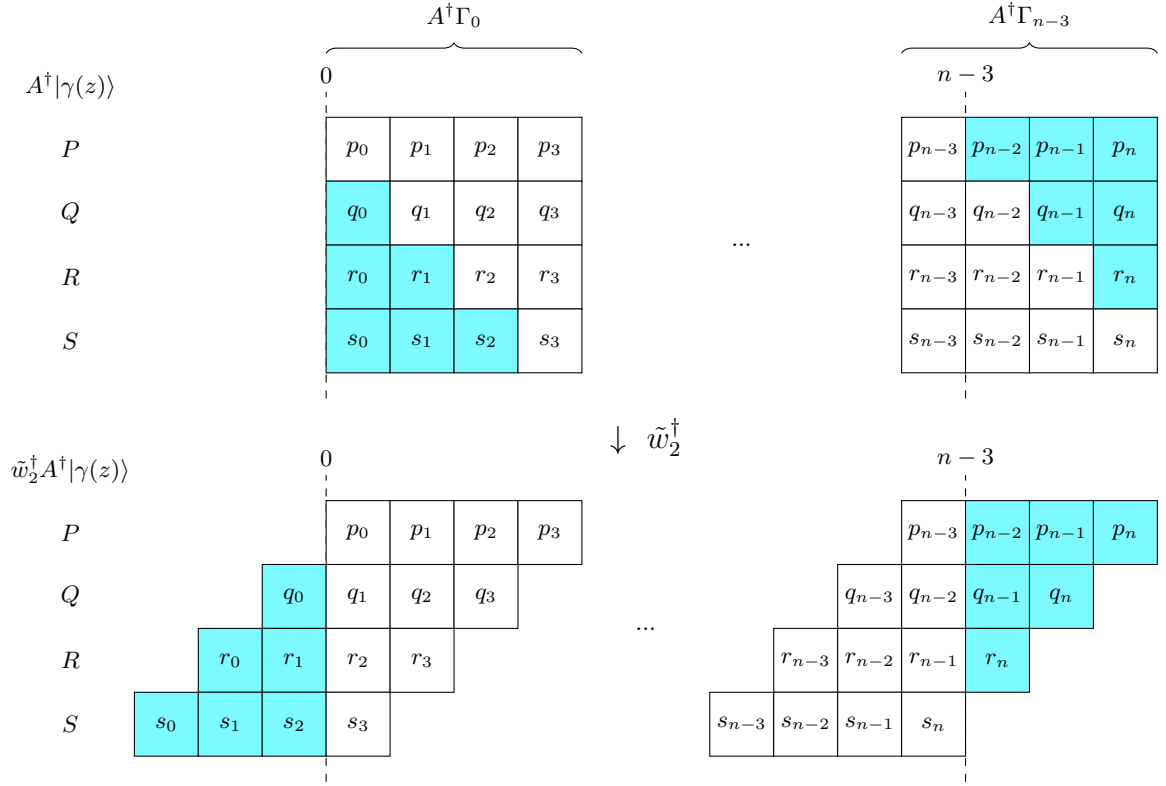
\noindent One can consider a polynomial state $\ket{\gamma}$ of degree $n$ to be of degree $n+d$ (whose last $d$ coefficients are zero), and then apply Lemma~\ref{thm:exponential-reduction} to $\Gamma_0$, $\Gamma_{n + d - (2^b - 1)}$. If $\Gamma_0^\dag \Gamma_{n + d - (2^b - 1)}$ is lower triangular, then Lemma~\ref{thm:exponential-reduction} gives us a unitary $A$ such that $w_b^\dag A^\dag \ket{\gamma}$ has degree $\le n + d - (2^b - 1)$. A surprising result comes as a corollary
\begin{lemmarestate}{\ref{thm:linear-reduction}}[Linear reduction]
    Any polynomial state $\ket{\gamma(z)}$ of degree $n$ admits a reduction of degree $\le n - 1$.
\end{lemmarestate}
\begin{proof}
    Consider $\ket{\gamma(z)}$ as a polynomial of degree $n + 2^b - 2$ (whose last $2^b - 2$ coefficients are zero).
    Its coefficient matrix $\Gamma_{n-1}$ has $\ket{\gamma_{n-1}}, \ket{\gamma_n}$ as the first two columns and zeroes everywhere else. By the definition of polynomial state
    \begin{align*}
        \braket{\gamma}{\gamma} = \sum_{k, \ell} \braket{\gamma_\ell}{\gamma_k} z^{k - \ell} = 1.
    \end{align*}
    Since this holds as a polynomial equation, every coefficient (except the zero-th one) has to be zero, in particular the $n$-th coefficient which gives $\braket{\gamma_0}{\gamma_n} = 0$. Therefore the product $\Gamma_0^\dag \Gamma_{n + 2^b - 2}$ satisfies:
    \begin{align*}
        \bra{0}\Gamma_0^\dag \Gamma_{n - 1}\ket{1} = \braket{\gamma_0}{\gamma_n} = 0
    \end{align*}
    along with the fact that all the other columns are zero (because they are zero in $\Gamma_{n-1}$), we conclude that $\Gamma_0^\dag \Gamma_{n-1}$ is lower triangular. By Lemma~\ref{thm:exponential-reduction}, this implies the existence of a unitary $A$ such that the reduction $\Tilde{w}^\dag_2 A^\dag \ket{\gamma}$ has degree $\le n + (2^b - 2) - (2^b - 1) = n - 1$.
\end{proof}
\noindent This argument is exactly the one used in usual proof of Theorem~\ref{thm:single-qubit-qsp-analytic} for single-qubit quantum signal processing~\cite{martynGrandUnificationQuantum2021, haahProductDecompositionPeriodic2019,motlaghGeneralizedQuantumSignal2023}, and it is also sufficient to prove
\begin{theoremrestate}{\ref{thm:multi-qubit-qsp-analytic-linear}}
    Let $\{ P_x \}_{0 \le x < 2^b} \in \C[z]$. There exist $A_0, \ldots, A_n \in SU(2^b)$ such that:
    \begin{align*}
        A_n \Tilde{w} A_{n-1} \Tilde{w} \cdots \Tilde{w} A_0 \ket{0} = \sum_x P_x(z) \ket{x}
    \end{align*}
    if and only if:
    \begin{enumerate}[(i)]
        \item $\deg P_x \le n$ for every $0 \le x < 2^b$;
        \item $\sum_x |P_x(z)|^2 = 1$ for every $z \in U(1)$.
    \end{enumerate}
\end{theoremrestate}
\begin{proof}
    For $n = 0$, the all the polynomials are constant and $\ket{\gamma}$ is a quantum state which can be implemented with $A_0 \ket{0}$ for a suitable choice of $A_0$ (also the sequence $\id, \ldots, \id, A_0$ implements the same polynomial state, for any $n$).
    For $n > 0$, by Theorem~\ref{thm:linear-reduction} there is a reduction $(A_n \Tilde{w})^\dag \ket{\gamma}$ of degree $n - 1$ for some $A_n \in SU(2^b)$. This reduction can be implemented with a $n - 1$ protocol $A_0, \ldots, A_{n-1}$, by induction. Hence, the sequence $A_0, \ldots, A_n$ implements $\ket{\gamma}$.
\end{proof}
\noindent While we proved that any polynomial state of degree $n$ can be implemented within $n$ steps, we can leverage the presence of higher powers of $z$ to reduce the degree by more than one during a step of the protocol, provided that some additional orthogonality conditions are given. The extremal case is when all the coefficient vectors are orthogonal.
\begin{theoremrestate}{\ref{thm:multi-qubit-qsp-analytic-one-length}}
    Let $\ket{\gamma(z)}$ be a $2^b$-dimensional polynomial state. Then there exist $A_0, A_1 \in SU(2^b)$ such that
    \begin{align*}
        A_1 \Tilde{w}_b A_0 \ket{0} = \ket{\gamma(z)}
    \end{align*}
    if and only if
    \begin{enumerate}[(i)]
        \item $P_x(z) := \braket{x}{\gamma(z)}$ has degree $\le 2^b - 1$ for every $0 \le x < 2^b$;
        \item $\sum_x |P_x(z)|^2 = 1$ for every $z \in U(1)$;
        \item $\Gamma_0^\dag \Gamma_0$ is diagonal, i.e., the coefficient vectors $\{ \ket{\gamma_k} \}_k$ are pairwise orthogonal.
    \end{enumerate}
\end{theoremrestate}
\begin{proof}
    If a polynomial state $\ket{\gamma}$ has degree $\le 2^b - 1$, then it can be reduced to a quantum state in a single step if and only if $\Gamma_0^\dag \Gamma_{2^b - 1}$ is lower triangular, by Lemma~\ref{thm:exponential-reduction}. Since $\Gamma_0^\dag \Gamma_{2^b - 1} = \Gamma_0^\dag \Gamma_0$, this matrix is also Hermitian and thus being lower triangular is equivalent to being diagonal. Conversely, $\ket{\gamma'} = \Tilde{w}_b A_0 \ket{0}$ has the $k$-th coefficient vector $\ket{\gamma'_k} \propto \ket{k}$ by construction, thus the coefficient vectors are pairwise orthogonal, and $A_1$ preserves this fact.
\end{proof}
\noindent From the proof one can see that the elements of the diagonal $\Gamma_0^\dag \Gamma_0$ are precisely the amplitudes squared of $A_0 \ket{0}$ (they sum up to $1$ because of condition~(ii)). In the meanwhile, the final $\ket{\gamma_k}$'s form an orthogonal basis which, if normalized, give the columns of $A_1$.

\end{document}